\documentclass[11pt]{article}
\usepackage[a4paper, total={6in, 8in}]{geometry}
\usepackage{enumerate}
\usepackage{amsmath}
\usepackage{amssymb,latexsym}
\usepackage{amsthm}
\usepackage{color}
\usepackage{cancel}
\usepackage{graphicx}
\usepackage{mathtools}
\usepackage{float}
\usepackage{hyperref}
\usepackage{comment}
\usepackage{amscd}
\usepackage{cite}
\usepackage{lineno}
\makeatletter
\usepackage{comment}
\let\wfs@comment@comment\comment
\let\comment\@undefined

\usepackage{changes}
\let\wfs@changes@comment\comment
\let\comment\@undefined

\newcommand\comment{%
    \ifthenelse{\equal{\@currenvir}{comment}}
    {\wfs@comment@comment}
    {\wfs@changes@comment}%
}

\definechangesauthor[name=Daniele, color=blue]{DAN}
\definechangesauthor[name=Giovanni, color=blue]{GIO}
\definechangesauthor[name=Ferdinando, color=blue]{FE}

\newtheorem{theorem}{Theorem}[section]

\newtheorem{proposition}[theorem]{Proposition}

\newtheorem{remark}[theorem]{Remark}

\newtheorem*{theorem*}{Theorem}

\newcommand\fq{\mathbb{F}_q}
\newcommand\F{\mathbb{F}}

\newcommand\fqn{\mathbb{F}_{q^n}}
\newcommand\Tr{\mathrm{Tr}}

\newcommand\C{\mathcal{C}}
\newcommand\cL{\mathcal{L}}

\title{Non-minimum tensor rank Gabidulin codes}
\author{Daniele Bartoli\thanks{Department of Mathematics and Informatics, University of Perugia, Perugia, Italy.
Email address: daniele.bartoli@unipg.it}\hspace{0.15 cm}  and Giovanni Zini\thanks{Department of Physics, Informatics and Mathematics, University of Modena and Reggio Emilia, Modena, Italy. Email address: giovanni.zini@unimore.it} \hspace{0.15 cm} and  Ferdinando Zullo\thanks{Department of Mathematics and Physics, University of Campania, Caserta, Italy. Email address: ferdinando.zullo@unicampania.it} }
\date{}
\begin{document}

\maketitle

\begin{abstract}
    The tensor rank of some Gabidulin codes of small dimension is investigated. In particular, we determine the tensor rank of any rank metric code equivalent to an $8$-dimensional $\fq$-linear generalized Gabidulin code in $\mathbb{F}_{q}^{4\times4}$. This shows that such a code is never minimum tensor rank.
    In this way, we detect the first infinite family of Gabidulin codes which are not minimum tensor rank.
\end{abstract}

\section{Introduction}

Rank metric codes were introduced by Delsarte \cite{de78} in 1978 and have been used in several contexts, such as crisscross error correction \cite{roth1991maximum}, cryptography \cite{gabidulin1991ideals}, and network coding \cite{silva2008rank}.
Because of their ubiquitous applications, they attracted increasing attention in the last years; see e.g. \cite{GR,PZ,sheekey_newest_preprint}.

Very recently, rank metric codes have been investigated through their tensor rank; see \cite{ByrneNeri,ByrneCotardo,ByrneCotardo2}.
Indeed, a rank metric code $\mathcal{C}$ in $\fq^{n\times m}$ can be seen as the \emph{slice space} of an associated \emph{generator $3$-tensor}, similarly to the case of linear codes in the Hamming metric, where a code can be described as the row space of a generator matrix. Therefore, after Byrne, Neri, Ravagnani and Sheekey \cite{ByrneNeri}, the \textbf{tensor rank} of $\C$ is defined as the tensor rank of a generator tensor of $\C$. 
Determining the tensor rank of a certain rank metric code is a hard problem in general and the exact value is known only for specific classes of codes; indeed the problem of computing the rank of a $3$-tensor is NP-complete \cite{NP}. Several lower and upper bounds for the tensor rank of a rank metric code were presented in \cite{ByrneNeri} and \cite{ByrneCotardo}. In particular, as a consequence of Kruskal's bound \cite{Kru}, the tensor rank of an $h$-dimensional $\fq$-linear rank metric code $\C$ in $\fq^{n\times m}$ of minimum distance $d$ is lower bounded by $h-d+1$.
The code $\C$ is said to be \textbf{minimum tensor rank} (\textbf{MTR} for short) if its tensor rank is exactly $h-d+1$.
The interest for rank metric codes with a low tensor rank is due to the following fact: the smaller the tensor rank of the generating tensors, the more efficient the encoding. 
Via the correspondence in \cite{delacruz} between full rank codes and semifields, the notion of tensor rank for rank metric codes extends the same notion for semifields, which was used as an invariant by Lavrauw in \cite{lavten}.
Moreover, some criteria by Kruskal \cite[Section 4]{Kru} use the rank of a tensor to assure its identifiability, i.e. the uniqueness of the pure tensors appearing in its decomposition, which is of interest for the numerical applications within statistics; see \cite{CO} and \cite[Section 2]{AMR}.

A family of particular interest among rank metric codes is the one of square Gabidulin codes $\mathcal{G}_{k,s}$ in $\fq^{n\times n}$, as they are \emph{maximum rank distance}, and indeed they have been deeply investigated.
However, their tensor rank is not known in general; exact results have been provided in \cite{ByrneNeri} and \cite{ByrneCotardo} when $k\in \{1,n-1\}$ and in few other cases.
Interestingly, when $q$ is large enough, Gabidulin codes with $k\in \{1,n-1\}$ turn out to be MTR codes.

In this paper we are interested in determining the tensor rank of those codes which are equivalent to an $\fq$-linear $8$-dimensional Gabidulin code in $\fq^{4\times 4}$.
The strategy that we apply makes use of \cite[Proposition 3.4]{ByrneNeri}, which involves rank-one matrices. The framework of our arguments is the one of linearized polynomials, where rank-one matrices correspond to trace functions of the shape $\alpha \Tr(\beta x)$, where $\Tr:\mathbb{F}_{q^4}\to\mathbb{F}_q$ and $\Tr(x)= x+x^q+x^{q^2}+x^{q^3}$.
Our main result is the following.
\begin{theorem}\label{Th:main3}
Let $q$ be a prime power, and $\C$ be a code which is equivalent to an $\fq$-linear $8$-dimensional generalized Gabidulin code in $\fq^{4\times 4}$.
Then the tensor rank of $\C$ is $11$ if $q\geq 3$, and $12$ if $q=2$.
In particular, $\C$ is not MTR.
\end{theorem}

The paper is organized as follows.
Section \ref{sec:prelim_codes} contains preliminary notions on rank metric codes and on the correspondence with linearized polynomials in the case of square codes.
Section \ref{sec:prelim_tensors} describes basic definitions and known results about tensors and the tensor rank of square generalized Gabidulin codes.
Section \ref{sec:ten} is devoted to the proof of Theorem \ref{Th:main3}: Section \ref{Sec:non10} shows that $\C$ is not MTR, while in Section \ref{Sec:11} we determine the tensor rank of $\C$ for $q\geq5$.
The remaining small values of $q$, are worked out computationally in Section \ref{Sec:particular}, as well as other Gabidulin codes in $\fq^{n\times n}$ with small values of $q$ and $n$.
Finally, the Appendix contains two auxiliary results which are needed in Section \ref{Sec:non10}, whose proof are quite technical.

\section{Rank metric codes and linearized polynomials}\label{sec:prelim_codes}

The set $\F_q^{n\times m}$ of matrices can be equipped with the \textbf{rank metric}, as
\[d(A,B) = \mathrm{rk}\,(A-B), \quad \mbox{ for } A,B\in\fq^{n\times m}.\]
 A \textbf{rank metric code}  is a subset $\C$  of $\F_q^{n \times m}$ endowed with the rank metric and its \textbf{minimum rank distance} is defined as
\[d := d(\C) =\min\{ d(A,B) \colon A,B \in \C,\,\, A\neq B \}.\]
Two $\fq$-linear rank metric codes $\C$ and $\C'$ in $\mathbb{F}_q^{n\times m}$ are \textbf{linearly equivalent} if and only if there exist $X \in \mathrm{GL}(n,q)$ and $Y \in \mathrm{GL}(m,q)$ such that
\[\C'=\{XC Y \colon C \in \C\},\]
or, if $m=n$,
\[\C'=\{XC^\top Y \colon C \in \C\},\]
where $C^\top$ denote the transpose of $C$.
Since in this paper we will only consider linear equivalence, we will refer to it simply as \emph{equivalence}.

Delsarte showed in \cite{de78} that the parameters of a rank metric code $\C$ satisfy a Singleton-like bound, namely
\[ |\C| \leq q^{\max\{m,n\}(\min\{m,n\}-d+1)}. \]
When equality holds, we call $\C$ a \textbf{maximum rank distance} (\textbf{MRD} for short) code.

In this paper we are interested only in the square case $m=n$, and in this case rank metric codes can be described in terms of linearized polynomials.
Indeed, consider the \textbf{$\fq$-linearized} (or simply \textbf{linearized}) polynomials of normalized degree over $\fqn$, i.e. elements of the form
$$ f(x)=\sum_{i=0}^{n-1}f_i x^{q^i}, \quad f_i \in \fqn.$$
The set of linearized polynomials is an $\fq$-algebra $\cL_{n,q}$ with the usual addition, the scalar multiplication by elements of $\fq$ and the composition modulo $x^{q^n}-x$.
It is well-known that the $\fq$-algebras $\cL_{n,q}$ and $\mathrm{End}_{\fq}(\fqn)$ are isomorphic,
via the correspondence between the linearized polynomial $f(x)$ and the $\fq$-endomorphism
$$ \alpha \longmapsto \sum_{i=0}^{n-1}f_i \alpha^{q^i}$$
of $\fqn$.
Hence, $\cL_{n,q}$ is also isomorphic to the $\fq$-algebra $\fq^{n\times n}$ of $n\times n$ matrices over $\fq$.
In this correspondence, the rank of a matrix in $\fq^{n\times n}$ equals the rank of the corresponding linearized polynomial in $\cL_{n,q}$ as an $\fq$-endomorphism of $\fqn$.
Therefore, rank metric codes in $\fq^{n\times n}$ can be seen as sets of linearized polynomials in $\cL_{n,q}$, so that we can speak of rank metric codes in $\cL_{n,q}$.
Notice that the set of matrices of rank $1$ in $\fq^{n\times n}$ corresponds to the set of elements of $\mathcal{L}_{n,q}$ of the shape $\alpha \Tr(\beta x)$ for some $\alpha,\beta \in \mathbb{F}_{q^n}^*$, where  $\Tr(z)= z+z^q+\cdots+z^{q^{n-1}}$; see \cite[Theorem 2.24]{LR}.
For a reference on linearized polynomials see \cite{wuliu}.

The first class of square MRD codes in the literature was the one of \textbf{generalized Gabidulin codes}, namely the $\fqn$-subspaces
\[ \mathcal{G}_{k,s}=\langle x,x^{q^s},\ldots, x^{q^{s(k-1)}} \rangle_{\fqn} \]
of $\cL_{n,q}$, where $1\leq k\leq n$ and $\gcd(s,n)=1$; they are MRD codes with $\fq$-dimension $kn$ and minimum distance $n-k+1$.
Gabidulin codes where first introduced by Delsarte in \cite{de78} and later by Gabidulin in \cite{gab} in the case $s=1$, and by Gabidulin and Kshevetskiy in \cite{gab2} in the general case.

\section{Tensor rank of generalized Gabidulin codes}\label{sec:prelim_tensors}

The tensors we will investigate in this paper are \textbf{$3$-tensors} in $\fq^h\otimes \fq^n\otimes\fq^m$.
If $\{u_1,\ldots,u_h\}$, $\{v_1,\ldots,v_n\}$, and $\{w_1,\ldots,w_n\}$ are bases of $\fq^h$, $\fq^n$, and $\fq^m$ respectively, then an $\fq$-basis of $\fq^h\otimes \fq^n\otimes\fq^m$ is given by 
\[ \{ u_l\otimes v_i\otimes w_j \colon 1 \leq l \leq h, 1 \leq i \leq n,1\leq j \leq m \}. \]
The tensors of the form $u\otimes v \otimes w$, with $u \in \fq^h$, $v\in\fq^n$ and $w \in \fq^m$, are called \textbf{simple} (or \textbf{pure}) \textbf{tensors}.
The \textbf{tensor rank} of a tensor $X\in \fq^h\otimes \fq^n\otimes\fq^m$ is defined as
\[ \mathrm{trk}(X)=\min\left\{ R\in \mathbb{N}_0\colon X=\sum_{i=1}^R u_i\otimes v_i\otimes w_i,\,\, u_i\in \fq^h, v_i\in\fq^n,w_i \in \fq^m \right\}. \]
Let $[i]=\{1,\ldots,i\}$. A $3$-tensor $X \in \fq^h\otimes \fq^n\otimes\fq^m$ can be represented as a map $X \colon [h]\times [n]\times [m]\rightarrow \fq$ given by $X=(X_{lij}\colon 1\leq l \leq h, 1\leq i\leq n,1\leq j\leq m)$. Therefore $\fq^h\otimes \fq^n\otimes\fq^m$ can be identified with the space $\fq^{h\times n\times m}$, and the tensor $X$ can be written as $X=(X_1,\ldots,X_h)$ with $X_i \in \fq^{n\times m}$.
The \textbf{first slice space} of $X$, denoted by $\mathrm{ss}_1(X)$, is the $\fq$-subspace of $\fq^{n\times m}$ generated by $X_1,\ldots,X_h$.
If $\dim_{\fq}(\mathrm{ss}_1(X))=h$, we say that $X$ is $1$-\textbf{nondegenerate}. 

The following result will be a key tool in our investigation.
\begin{proposition}(see \cite[Proposition 3.4]{ByrneNeri} and \cite[Proposition 14.45]{AlgebraicComplexity})\label{prop:tenrank}
Let $X\in \fq^{h\times n\times m}$ and $R$ be a positive integer. The following are equivalent:
\begin{enumerate}
    \item $\mathrm{trk}(X)\leq R$;
    \item there exist $A_1,\ldots,A_R \in \fq^{n\times m}$ of rank $1$ such that $\mathrm{ss}_1(X)\subseteq \langle A_1,\ldots,A_R \rangle_{\fq}$.
\end{enumerate}
In particular, $\mathrm{trk}(X)= R$ if and only if $R$ is the minimum integer such that there exist $A_1,\ldots,A_R \in \fq^{n\times m}$ of rank $1$ satisfying $\mathrm{ss}_1(X)\subseteq \langle A_1,\ldots,A_R \rangle_{\fq}$.
\end{proposition}

Kruskal in \cite{Kru} bounded the tensor rank of a $3$-tensor, using the following map:
\[ m_1\colon \fq^{s\times h}\times \fq^{h\times n\times m}\rightarrow \fq^{s\times n\times m}, \quad \left(A,\sum_i u_i\otimes v_i\otimes w_i\right)\mapsto \sum_i(Au_i)\otimes v_i\otimes w_i.  \]
\begin{theorem}(see \cite[Corollary 1]{Kru})\label{th:Kru}
Let $X\in \fq^{h\times n\times m}$ be $1$-nondegenerate, then
\[ \mathrm{trk}(X)\geq h+\min\{ \mathrm{trk}(m_1(u,X)) \colon u \in \fq^h\setminus\{0\} \}-1. \]
\end{theorem}

Tensors are related to rank metric codes as follows. Let $\C$ be an $\fq$-linear code in $\fq^{n\times m}$ of dimension $h$ and minimum distance $d$. A \textbf{generator tensor} for $\C$ is a $3$-tensor $X\in \fq^{h\times n\times m}$ such that $\mathrm{ss}_1(X)=\C$.
As proved in \cite[Proposition 4.2]{ByrneNeri}, two generator tensors of the same rank metric code $\C$ have the same tensor rank. Therefore, we can define the \textbf{tensor rank} $\mathrm{trk}(\C)$ of $\C$ as the tensor rank of any generator tensor of $\C$.

\begin{proposition}(see \cite[Proposition 4.5]{ByrneNeri})\label{prop:equiv}
If $\C,\C^{\prime}$ are equivalent codes, then $\mathrm{trk}(\C)=\mathrm{trk}(\C^{\prime})$.
\end{proposition}

By Theorem \ref{th:Kru},
\begin{equation}\label{eq:bound} \mathrm{trk}(\C)\geq h+d-1. \end{equation}
If $\C$ attains equality in \eqref{eq:bound}, it is called a \textbf{minimum tensor rank} (\textbf{MTR} for short) code.

Although Gabidulin codes form the most studied family of rank metric codes, the complete determination of their tensor rank is still missing.
We now describe the known results on the tensor rank of square Gabidulin codes $\mathcal{G}_{k,s}\subset\mathcal{L}_{n,q}$.
Bound \eqref{eq:bound} reads as follows.
\begin{theorem}
For every $k\leq n$, we have $\mathrm{trk}(\mathcal{G}_{k,s})\geq(k+1)n-k$.
\end{theorem}

The tensor rank of $\mathcal{G}_{1,s}$ coincides with the tensor rank of the field $\fqn$ (see \cite{Knuth} and \cite{Liebler} where semifields were described for the first time in terms of tensors). 
By \cite[Propositions 14.47 and 14.48]{AlgebraicComplexity} and a link with a well-studied tensor pointed out in \cite[Lemma 5.13]{ByrneNeri}, it follows that $\mathrm{trk}(\mathcal{G}_{1,s})=2n-1$ if $q\geq 2n-1$, and $\mathrm{trk}(\mathcal{G}_{1,s})>2n-1$ if $q\leq 2n-2$.
For $n=3$, $\mathrm{trk}(\mathcal{G}_{1,s})=6$ if $q \in \{2,3\}$ 
(see \cite[Lemma 15]{LavrauwSheekey} and also \cite{LPZ}).
For $n=4$, $\mathrm{trk}(\mathcal{G}_{1,s})=9$ if $q\in \{2,3\}$ (see \cite[Example 6.4]{ByrneNeri} for $q=2$ and \cite[Theorem 4]{LavrauwSheekey} for $q=3$), while $\mathrm{trk}(\mathcal{G}_{1,s})$ is unknown for $q\in\{4,5\}$.
Further bounds and asymptotic results for the tensor rank of $\fqn$ are known, see e.g.\ \cite{surveyTensorRank}. 

The following upper bound follows from  the tensor rank of $\mathcal{G}_{1,s}$.

\begin{theorem}(see \cite[Proposition 5.15]{ByrneNeri})
Let $q\geq 2n-2$. For every $k\leq n$, we have $\mathrm{trk}(\mathcal{G}_{1,s})\leq\min\{n^2, k(2n - 1)\}$.
\end{theorem}
A partial result is known also in the case of Gabidulin codes  $\mathcal{G}_{n-1,s}$.
\begin{theorem}(see \cite[Theorem 5.15]{ByrneCotardo})
Let $q\geq n$. Then $\mathrm{trk}(\mathcal{G}_{n-1,s})=n^2-n+1$.
\end{theorem}
The tensor rank of Gabidulin codes $\mathcal{G}_{k,s}$ with $k \notin \{1,n-1\}$ is not known.
In this paper we study the first open case, namely $k=2$ and $n=4$.
In Section \ref{Sec:particular} we will investigate the remaining open cases when $n\leq 4$.

\section{The tensor rank of $\mathcal{G}_{2,1}\subset\mathcal{L}_{4,q}$}\label{sec:ten}

The two $8$-dimensional generalized Gabidulin codes $\mathcal{G}_{2,1}$ and $\mathcal{G}_{2,3}$ in $\mathcal{L}_{4,q}$ are easily seen to be equivalent. Therefore, by Proposition \ref{prop:equiv}, in order to prove Theorem \ref{Th:main3} it is enough to prove it for the Gabidulin code $\mathcal{G}=\mathcal{G}_{2,1}=\langle x,x^q\rangle_{\mathbb{F}_{q^4}}$.
In Section \ref{Sec:non10} we show that the tensor rank of $\mathcal{G}$ is not $10$ for any $q$.
In Section \ref{Sec:11} we prove that the tensor rank of $\mathcal{G}$ is $11$ if $q\geq5$. We complete the proof in Section \ref{Sec:particular}, where we determine the tensor rank of some Gabidulin codes for some values of $q$.

\subsection{The tensor rank of $\mathcal{G}$ is larger than $10$}\label{Sec:non10}

This section is devoted to the proof of the following theorem.
\begin{theorem}\label{Th:main}
For any prime power $q$, we have $\mathrm{trk}(\mathcal{G})\geq11$. Thus, $\mathcal{G}$ is not an MTR code.
\end{theorem}

By Proposition \ref{prop:tenrank} and Section \ref{sec:prelim_codes}, $\mathrm{trk}(\mathcal{G})=10$ if and only if there exist $10$ trace functions $\alpha_i \Tr(\beta_i x)$ such that $\mathcal{G}\subseteq \langle \alpha_1\Tr(\beta_1 x),\ldots, \alpha_{10}\Tr(\beta_{10} x)\rangle_{\mathbb{F}_{q}}$.
This is equivalent to say that there exist $\alpha_1\Tr(\beta_1 x),\alpha_2\Tr(\beta_2 x)$ such that there exists an $\mathbb{F}_q$-basis of $ \langle x,x^q\rangle_{\F_{q^4}}\oplus\langle \alpha_1\Tr(\beta_1 x),\alpha_2\Tr(\beta_2 x)\rangle_{\mathbb{F}_{q}}$ only composed by traces.
So, consider $\alpha_1,\alpha_2,\beta_1,\beta_2 \in \F_{q^4}$ such that $H= \langle x,x^q\rangle_{\F_{q^4}}\oplus\langle \alpha_1\Tr(\beta_1 x),\alpha_2\Tr(\beta_2 x)\rangle_{\mathbb{F}_{q}}$ has dimension $10$ over $\fq$.

The proof strategy relies on two steps:\\
\textbf{Step 1:} To find explicit necessary and sufficient conditions on $\alpha_3,\beta_3 \in \F_{q^4}$ such that $\alpha_3\Tr(\beta_3 x) \in H$.\\
\textbf{Step 2:} To prove the non-existence of ten $\fq$-linearly independent traces in $H$.

In Steps 1 and 2 we will also need auxiliary results (Theorems \ref{Th:O'Sistemone} and \ref{Th:Ranghi} respectively) which are in Appendix, due to their technicality.

\begin{proof}
\textbf{Step 1:} 
Suppose that $\alpha_3 \Tr(\beta_3 x)$ is in $H$. Then there exist $\gamma,\delta\in \mathbb{F}_{q^4}$, $c_1,c_2\in \mathbb{F}_q$ such that
\begin{eqnarray*}
\alpha_3 \Tr(\beta_3 x)&=&\gamma x+\delta x^q+c_1\alpha_1 \Tr(\beta_1 x)+c_2\alpha_2 \Tr(\beta_2 x).
\end{eqnarray*}
This polynomial identity implies that $\gamma$, $\delta$, $c_1$, and $c_2$ satisfy the following system:
\begin{equation}\label{System1}
\begin{cases}
\gamma+c_1\alpha_1\beta_1+c_2\alpha_2 \beta_2&=\alpha_3\beta_3,\\
\delta+c_1\alpha_1\beta_1^q+c_2\alpha_2 \beta_2^q&=\alpha_3\beta_3^q,\\
c_1\alpha_1\beta_1^{q^2}+c_2\alpha_2 \beta_2^{q^2}&=\alpha_3\beta_3^{q^2},\\
c_1\alpha_1\beta_1^{q^3}+c_2\alpha_2 \beta_2^{q^3}&=\alpha_3\beta_3^{q^3}.\\
\end{cases}
\end{equation}
It cannot happen that $\beta_2/\beta_1$  and $\alpha_2/\alpha_1$ are both in $\fq$. Indeed, if $\beta_2=\rho \beta_1$ and $\alpha_2=\omega \alpha_1$ for some $\rho,\omega \in \fq$ then 
\[ \alpha_2 \Tr(\beta_2 x)=\rho\omega \alpha_1\Tr(\beta_1 x) \]
and hence $\dim_{\fq}(H)<10$.

Note that if $c_1=0$, then the last two equations of System  \eqref{System1} yield  $c_2\alpha_2/\alpha_3=(\beta_3/\beta_2)^{q^2}=(\beta_3/\beta_2)^{q^3}$ which implies $\beta_3/\beta_2\in \mathbb{F}_q^*$ and $\alpha_3/\alpha_2\in \mathbb{F}_q^*$. This means that $\alpha_3 \Tr(\beta_3 x)=\lambda \alpha_2 \Tr(\beta_2 x)$ for some $\lambda \in \mathbb{F}_q$. A similar conclusion arises from $c_2=0$. 
So, in these cases $\alpha_3 \Tr(\beta_3 x)\in \langle \alpha_1 \Tr(\beta_1 x),\alpha_2 \Tr(\beta_2 x)\rangle_{\fq}$.

From now on we always assume $c_1c_2\neq0$. By the last two equations in System \eqref{System1} one gets
\begin{equation}\label{eq:beta}\beta_3^{q-1}=\frac{c_1\alpha_1^{q^2}\beta_1^{q}+c_2\alpha_2^{q^2} \beta_2^{q}}{c_1\alpha_1^{q^2}\beta_1+c_2\alpha_2^{q^2} \beta_2}=\frac{\frac{c_1}{c_2}\alpha_1^{q^2}\beta_1^{q}+\alpha_2^{q^2} \beta_2^{q}}{\frac{c_1}{c_2}\alpha_1^{q^2}\beta_1+\alpha_2^{q^2} \beta_2}=:A(c_1/c_2).\end{equation}
Note that, since $\beta_3\neq0$, $c_1\alpha_1^{q^2}\beta_1^{q}+c_2\alpha_2^{q^2} \beta_2^{q}=0$ if and only if $c_1\alpha_1^{q^2}\beta_1+c_2\alpha_2^{q^2} \beta_2=0$, that is $\beta_1/\beta_2$ and $\alpha_1/\alpha_2$ both belong to $\mathbb{F}_q^*$. 
Indeed, since $c_1$ and $c_2$ are both nonzero,  
\[ \det\begin{pmatrix}
\alpha_1^{q^2}\beta_1^q & \alpha_2^{q^2}\beta_2^q \\
\alpha_1^{q^2}\beta_1 & \alpha_2^{q^2}\beta_2
\end{pmatrix}=0,
\]
which implies $\beta_1/\beta_2 \in \fq^*$.
Similarly, one can show that $\alpha_1/\alpha_2 \in \fq^*$.
However, this is a contradiction to our assumptions. 

An element $\beta_3\in \mathbb{F}_{q^4}$ satisfying Equation \eqref{eq:beta} exists if and only if $A(c_1/c_2)^{1+q+q^2+q^3}=1$, that is,
\begin{eqnarray}(c_1\alpha_1^{q^2}\beta_1^{q}+c_2\alpha_2^{q^2} \beta_2^{q})
(c_1\alpha_1^{q^3}\beta_1^{q^2}+c_2\alpha_2^{q^3} \beta_2^{q^2})
(c_1\alpha_1\beta_1^{q^3}+c_2\alpha_2 \beta_2^{q^3})
(c_1\alpha_1^{q}\beta_1+c_2\alpha_2^{q} \beta_2)
=\nonumber\\
(c_1\alpha_1^{q^2}\beta_1+c_2\alpha_2^{q^2} \beta_2)
(c_1\alpha_1^{q^3}\beta_1^{q}+c_2\alpha_2^{q^3} \beta_2^{q})
(c_1\alpha_1\beta_1^{q^2}+c_2\alpha_2 \beta_2^{q^2})
(c_1\alpha_1^{q}\beta_1^{q^3}+c_2\alpha_2^{q} \beta_2^{q^3}).\label{eq:NUM-DEN}
\end{eqnarray}

We are interested in bounding the number of non-$\mathbb{F}_q$-proportional pairs $(c_1,c_2)$, with $c_1c_2\neq0$. The above homogeneous polynomial in $c_1$ and $c_2$ is of degree at most three in both $c_1$ and $c_2$, 
and its coefficients are as follows:
\begin{itemize}
    \item[i)] the coefficient of $c_2^4$ is zero;
    \item[ii)] the coefficient of $c_1c_2^3$ is
    \[
    -\alpha_1\alpha_2^{q+q^2+q^3}\beta_1^{q^2}\beta_2^{1+q+q^3} + \alpha_1 \alpha_2^{q+q^2+q^3}\beta_1^{q^3}\beta_2^{1+q+q^2} + \alpha_1^{q}\alpha_2^{1+q^2+q^3}\beta_1\beta_2^{q+q^2+q^3} \]\[- \alpha_1^q \alpha_2^{1+q^2+q^3}\beta_1^{q^3}\beta_2^{1+q+q^2} - \alpha_1^{q^2}\alpha_2^{1+q+q^3}\beta_1\beta_2^{q+q^2+q^3} + \alpha_1^{q^2}\alpha_2^{1+q+q^3}\beta_1^{q}\beta_2^{1+q^2+q^3} \]\[- \alpha_1^{q^3}\alpha_2^{1+q+q^2}\beta_1^q\beta_2^{1+q^2+q^3} + \alpha_1^{q^3}\alpha_2^{1+q+q^2}\beta_1^{q^2}\beta_2^{1+q+q^3},
    \]
    \item[iii)] the coefficient of $c_1^2c_2^2$ is
    \[
    \alpha_1^{1+q}\alpha_2^{q^2+q^3}\beta_1^{1+q^3}\beta_2^{q+q^2} - \alpha_1^{1+q}\alpha_2^{q^2+q^3}\beta_1^{q^2+q^3}\beta_2^{1+q}  - \alpha_1^{1+q^2}\alpha_2^{q+q^3}\beta_1^{1+q^2}\beta_2^{q+q^3}\] \[+ \alpha_1^{1+q^2}\alpha_2^{q+q^3}\beta_1^{q+q^3}\beta_2^{1+q^2} - \alpha_1^{1+q^3}\alpha_2^{q+q^2}\beta_1^{q+q^2}\beta_2^{1+q^3} \] \[ + \alpha_1^{1+q^3}\alpha_2^{q+q^2}\beta_1^{q^2+q^3}\beta_2^{1+q} + \alpha_1^{q+q^2}\alpha_2^{1+q^3}\beta_1^{1+q}\beta_2^{q^2+q^3} \]\[- \alpha_1^{q+q^2}\alpha_2^{1+q^3}\beta_1^{1+q^3}\beta_2^{q+q^2} + \alpha_1^{q+q^3}\alpha_2^{1+q^2}\beta_1^{1+q^2}\beta_2^{q+q^3} \]\[- \alpha_1^{q+q^3}\alpha_2^{1+q^2}\beta_1^{q+q^3}\beta_2^{1+q^2} - \alpha_1^{q^2+q^3}\alpha_2^{1+q}\beta_1^{1+q}\beta_2^{q^2+q^3}+ \alpha_1^{q^2}\alpha_1^{q^3}\alpha_2^{1+q} \beta_1^{q+q^2}\beta_2^{1+q^3},
    \]
    \item[iv)] the coefficient of $c_1^3c_2$ is
    \[ \alpha_1^{1+q+q^2}\alpha_2^{q^3}\beta_1^{1+q+q^3}\beta_2^{q^2} - \alpha_1^{1+q+q^2} \alpha_2^{q^3}\beta_1^{1+q^2+q^3}\beta_2^q + \alpha_1^{1+q+q^3}\alpha_2^{q^2}\beta_1^{1+q^2+q^3}\beta_2^q  \]\[- \alpha_1^{1+q+q^3}\alpha_2^{q^2}\beta_1^{q+q^2+q^3}\beta_2 - \alpha_1^{1+q^2+q^3}\alpha_2^q \beta_1^{1+q+q^2}\beta_2^{q^3} + \alpha_1^{1+q^2+q^3}\alpha_2^q \beta_1^{q+q^2+q^3}\beta_2 \]\[ + \alpha_1^{q+q^2+q^3}\alpha_2\beta_1^{1+q+q^2}\beta_2^{q^3} - \alpha_1^{q+q^2+q^3}\alpha_2\beta_1^{1+q+q^3}\beta_2^{q^2}, \]
    \item[v)] the coefficient of $c_1^4$ is zero.
\end{itemize}

Therefore the number of non-$\fq$-proportional solutions $(c_1,c_2)$ with $c_1c_2\ne0$ is at most $2$, if the polynomial is non-vanishing.
Moreover, this polynomial vanishes if and only if
\begin{equation}\label{Eq:Sistemone}
\begin{cases}YZ^q - YZ^{q^2} + Y^q Z^{q^2} - Y^q  Z^{q^3} - Y^{q^2} Z + Y^{q^2}  Z^{q^3} +  Y^{q^3} Z -  Y^{q^3} Z^q=0\\
\\
Y^{q+1}  Z^{q^2+q} - Y^{q+1} Z^{q^3+q^2} -Y^{q^2+1} Z^{q^2+1} +  Y^{q^2+1} Z^{q^3+q}\\ + Y^{q^3+1} Z^{q+1}
      - Y^{q^3+1} Z^{q^2+q}      
 - Y^{q^2+q}  Z^{q^3+1} + Y^{q^2+q}Z^{q^3+q^2} +  Y^{q^3+q} Z^{q^2+1}\\ -
            Y^{q^3+q}  Z^{q^3+q} - Y^{q^3+q^2} Z^{q+1} + Y^{q^3+q^2}  Z^{q^3+1}=0\\
\\
Y^{q^2+q+1} Z^{q^3+q^2+1} - Y^{q^2+q+1} Z^{q^3+q^2+q} -  Y^{q^3+q+1} Z^{q^2+q+1} +
            Y^{q^3+q+1} Z^{q^3+q^2+q}\\ +  Y^{q^3+q^2+1} Z^{q^2+q+1} - Y^{q^3+q^2+1} Z^{q^3+q+1} +
            Y^{q^3+q^2+q}Z^{q^3+q+1} -  Y^{q^3+q^2+q} Z^{q^3+q^2+1}=0,\\
\end{cases}
\end{equation}
where $Y=\beta_1/\beta_2$ and $Z=\alpha_1/\alpha_2$.
The solutions $(Y,Z)\in \mathbb{F}_{q^4}$ of System \eqref{Eq:Sistemone} are given in Theorem \ref{Th:O'Sistemone}.
From now on we will suppose that $Y=\beta_1/\beta_2$ and $Z=\alpha_1/\alpha_2$ are solutions of System \eqref{Eq:Sistemone}. In this case, by Equation \eqref{eq:beta}, the maximum number of non-$\mathbb{F}_q$-proportional possible values of $\beta_3\in \mathbb{F}_{q^4}$ is $q-1$ when $\frac{c_1}{c_2}$ runs in $\fq^*$. By System \eqref{System1}, to each such value of $\beta_3$ there corresponds at most one value of $\alpha_3\in \mathbb{F}_{q^4}$ .

Define $\lambda=c_1/c_2\in \mathbb{F}_q^*$, so that $\beta_3=\beta_3(\lambda)$ satisfies  \begin{equation}\label{eq:condbeta3q-1}
\beta_3^{q-1}(\lambda)=\frac{\lambda\alpha_1^{q^2}\beta_1^{q}+\alpha_2^{q^2} \beta_2^{q}}{\lambda\alpha_1^{q^2}\beta_1+\alpha_2^{q^2} \beta_2}.
\end{equation}

Now, let $N(\lambda)=\lambda Z^{q^2} Y^{q}+1$ and $D(\lambda)=\lambda Z^{q^2} Y+1$, so that $\beta_3^{q-1}(\lambda)=\beta_2^{q-1}N(\lambda)/D(\lambda)$ and, by the third equation of System \eqref{System1},
\begin{equation}\label{eq:condbeta3q-12}
\alpha_3\beta_3=c_2\alpha_2\beta_2^{q^2}D^{q^2}(\lambda)/\beta_3^{q^2-1}=c_2\alpha_2\beta_2D^{q^2+q+1}(\lambda)/N^{q+1}(\lambda),
\end{equation}
and $\alpha_3\Tr(\beta_3 x)$ reads
\begin{eqnarray*}
&&\alpha_3\beta_3(x+\beta_3^{q-1}x^q+\beta_3^{q^2-1}x^{q^2}+\beta_3^{q^3-1}x^{q^3})\\
&&=\frac{c_2\alpha_2\beta_2D^{q^2+q+1}(\lambda)}{N^{q+1}(\lambda)}\left(x+\beta_2^{q-1}\frac{N(\lambda)}{D(\lambda)}x^q+\beta_2^{q^2-1}\frac{N^{q+1}(\lambda)}{D^{q+1}(\lambda)}x^{q^2}+\beta_2^{q^3-1}\frac{N^{q^2+q+1}(\lambda)}{D^{q^2+q+1}(\lambda)}x^{q^3} \right)\\
&&=c_2\left(\alpha_2\beta_2\frac{D^{q^2+q+1}(\lambda)}{N^{q+1}(\lambda)}x+\alpha_2\beta_2^q\frac{D^{q^2+q}(\lambda)}{N^{q}(\lambda)}x^q+\alpha_2\beta_2^{q^2}D^{q^2}(\lambda)x^{q^2}+\alpha_2\beta_2^{q^3}N^{q^2}(\lambda)x^{q^3} \right).
\end{eqnarray*}

\textbf{Step 2:} We prove that there exist no eight distinct values $\lambda_1,\ldots,\lambda_8\in \mathbb{F}_q^*$ such that the eight rank-one linear functions $F_i(x)= \alpha_2\beta_2\frac{D^{q^2+q+1}(\lambda_i)}{N^{q+1}(\lambda_i)}x+\alpha_2\beta_2^q\frac{D^{q^2+q}(\lambda_i)}{N^{q}(\lambda_i)}x^q+\alpha_2\beta_2^{q^2}D^{q^2}(\lambda_i)x^{q^2}+\alpha_2\beta_2^{q^3}N^{q^2}(\lambda_i)x^{q^3}$ and $F_9(x):=\alpha_1\Tr(\beta_1 x)$, $F_{10}(x):=\alpha_2\Tr(\beta_2 x)$ are $\mathbb{F}_q$-linearly independent. By Proposition \ref{prop:tenrank}, this will yield that $\mathrm{trk}(\mathcal{G})\geq11$.

Equivalently, we prove the existence of $\mu_1, \ldots,\mu_{10}\in \mathbb{F}_q$ such that
$\mu_1F_1(x)+\cdots+\mu_{10}F_{10}(x)=0$ and
not all the $\mu_i$'s are zero, so that the ten traces $F_i(x)$, $i=1,\ldots,10$, are $\fq$-linearly dependent.
Let $\mu_1, \ldots,\mu_{10}\in \mathbb{F}_q$ be such that
\begin{equation}\label{eq:condpoldeg}
\mu_1F_1(x)+\cdots+\mu_{10}F_{10}(x)=0. 
\end{equation}
In particular, Equation \eqref{eq:condpoldeg} can be seen as a polynomial identity; the coefficients of degree $q^3$ and $q^2$ yield to
$$\alpha_2 \beta_2^{q^3}\left(\sum_{i=1}^{8}\mu_i N^{q^2}(\lambda_i)+\mu_{10}\right)+\alpha_1 \beta_1^{q^3}\mu_{9}=0=\alpha_2 \beta_2^{q^2}\left(\sum_{i=1}^{8}\mu_i D^{q^2}(\lambda_i)+\mu_{10}\right)+\alpha_1 \beta_1^{q^2}\mu_{9}.$$
Since $N(\lambda_i)=\lambda_i Z^{q^2}Y^q+1$ and $D(\lambda_i)=\lambda_i Z^{q^2} Y+1$, 
$$\sum_{i=1}^{8}\mu_i +\mu_{10}=-\left(\sum_{i=1}^{8}\mu_i\lambda_i +\mu_9\right)Y^{q^3}Z  \textrm{ and } \sum_{i=1}^{8}\mu_i +\mu_{10}=-\left(\sum_{i=1}^{8}\mu_i\lambda_i +\mu_9\right)Y^{q^2}Z.$$

Suppose that $Y^{q^3}Z=Y^{q^2}Z$, which is equivalent to $Y\in\fq$. Then $N(\lambda)=D(\lambda)$ and $\beta_3(\lambda)^{q-1}=\beta_2^{q-1}=\beta_1^{q-1}$.
By System \eqref{System1}, this implies $\alpha_3 \Tr(\beta_3 x)=c_1 \alpha_1 \Tr(\beta_1 x)+c_2 \alpha_2 \Tr(\beta_2 x)\in\langle \alpha_1 \Tr(\beta_1 x),\alpha_2 \Tr(\beta_2 x)\rangle_{\fq}$ and hence the $F_i$'s, $i=1,\ldots,10$, are linearly dependent.

We can then assume that $Y^{q^3}Z\ne Y^{q^2}Z$, so that
\begin{equation}\label{Eq:primeDue}
    \mu_{9}+\sum_{i=1}^{8}\mu_i\lambda_i=\sum_{i=1}^{8}\mu_i+\mu_{10}=0.
\end{equation} 
Also, by looking at the the coefficients of degree $1$ and $q$ in Equation \eqref{eq:condpoldeg},
\begin{equation}\label{Eq:ultimeDue}
    \alpha_2\beta_2\sum_{i=1}^{8}\mu_i \frac{D^{q^2+q+1}(\lambda)}{N^{q+1}(\lambda)}+\mu_9\alpha_1\beta_1+\mu_{10}\alpha_2\beta_2=0=\alpha_2\beta_2^q\sum_{i=1}^{8}\mu_i \frac{D^{q^2+q}(\lambda)}{N^{q}(\lambda)}+\mu_9\alpha_1\beta_1^q+\mu_{10}\alpha_2\beta_2^q.
\end{equation}

Equations \eqref{Eq:ultimeDue} and their images under the $q$-Frobenius map, together with Equations \eqref{Eq:primeDue}, form a homogeneous linear system of ten equations whose matrix is 
\begin{equation}\label{eq:matrice}
M=\begin{pmatrix}
1&1&\cdots &0&1\\
\lambda_1&\lambda_2&\cdots&1&0\\
\frac{D^{q^2+q+1}(\lambda_1)}{N^{q+1}(\lambda_1)}&\frac{D^{q^2+q+1}(\lambda_2)}{N^{q+1}(\lambda_2)}&\cdots&YZ&1\\
\frac{D^{q^3+q^2+q}(\lambda_1)}{N^{q^2+q}(\lambda_1)}&\frac{D^{q^3+q^2+q}(\lambda_2)}{N^{q^2+q}(\lambda_2)}&\cdots&Y^qZ^q&1\\
\frac{D^{q^3+q^2+1}(\lambda_1)}{N^{q^3+q^2}(\lambda_1)}&\frac{D^{q^3+q^2+1}(\lambda_2)}{N^{q^3+q^2}(\lambda_2)}&\cdots&Y^{q^2}Z^{q^2}&1\\
\frac{D^{q^3+q+1}(\lambda_1)}{N^{q^3+1}(\lambda_1)}&\frac{D^{q^3+q+1}(\lambda_2)}{N^{q^3+1}(\lambda_2)}&\cdots&Y^{q^3}Z^{q^3}&1\\
\frac{D^{q^2+q}(\lambda_1)}{N^{q}(\lambda_1)}&\frac{D^{q^2+q}(\lambda_2)}{N^{q}(\lambda_2)}&\cdots&Y^{q}Z&1\\
\frac{D^{q^3+q^2}(\lambda_1)}{N^{q^2}(\lambda_1)}&\frac{D^{q^3+q^2}(\lambda_2)}{N^{q^2}(\lambda_2)}&\cdots&Y^{q^2}Z^{q}&1\\
\frac{D^{q^3+1}(\lambda_1)}{N^{q^3}(\lambda_1)}&\frac{D^{q^3+1}(\lambda_2)}{N^{q^3}(\lambda_2)}&\cdots&Y^{q^3}Z^{q^2}&1\\
\frac{D^{q+1}(\lambda_1)}{N(\lambda_1)}&\frac{D^{q+1}(\lambda_2)}{N(\lambda_2)}&\cdots&YZ^{q^3}&1\\
\end{pmatrix}.
\end{equation}

Since the rows of $M$ form orbits under the $q$-Frobenius map, the solutions of the associated system have entries in $\mathbb{F}_q$.
By Theorem \ref{Th:Ranghi}, the rank of $M$ is smaller than $10$.
Therefore, there are non-trivial solutions $(\mu_1,\ldots,\mu_{10})\in \mathbb{F}_q^{10}$ of $\mu_1F_1(x)+\cdots+\mu_{10}F_{10}(x)=0$. Then $\langle F_1(x),\ldots,F_{10}(x)\rangle_{\fq}$ has dimension smaller than $10$. This shows that $\mathrm{trk}(\mathcal{G})\geq11$. Thus, Theorem \ref{Th:main} is proved.
\end{proof}

\subsection{The tensor rank of $\mathcal{G}$ is $11$ for $q\geq5$}\label{Sec:11}

In this section we use the notations of Section \ref{Sec:non10}, and assume that $q\geq5$.
By Theorem \ref{Th:main}, $\mathrm{trk}(\mathcal{G})\geq11$.
We prove the following theorem.

\begin{theorem}\label{Th:main2}
For any prime power $q\geq 5$, we have $\mathrm{trk}(\mathcal{G})=11$.
\end{theorem}

By Proposition \ref{prop:tenrank}, it is enough to show the existence of $11$ $\fq$-linearly independent trace functions whose $\fq$-span contains $\mathcal{G}$.
Our key tool is Step 1 in Section \ref{Sec:non10}.

\begin{proof}
Let $\alpha_0,\beta_0\in\mathbb{F}_{q^4}^*$ and $\lambda_1,\lambda_1^{\prime},\ldots,\lambda_4,\lambda_4^{\prime}\in\mathbb{F}_q^*$ with $\lambda_i\ne\lambda_j$ and $\lambda_i^{\prime}\ne\lambda_j^{\prime}$ for $i\ne j$.
Let $\alpha,\beta,\alpha^{\prime},\beta^{\prime}\in\mathbb{F}_{q^4}^*$ be such that $Y:=\beta/\beta_0$, $Z:=\alpha/\alpha_0$, $Y^{\prime}:=\beta^{\prime}/\beta_0$ and $Z^{\prime}:=\alpha^{\prime}/\alpha_0$ satisfy $Z,Z^{\prime}\notin\mathbb{F}_{q^2}$ and $Y=1/Z^{q^2+q}$, $Y^{\prime}=1/(Z^{\prime})^{q^2+q}$.
By Theorem \ref{Th:O'Sistemone}, $(Y,Z)$ and $(Y^{\prime},Z^{\prime})$ are solutions of System \eqref{Eq:Sistemone}. As in the proof of Theorem \ref{Th:main}, by \eqref{eq:condbeta3q-1} and \eqref{eq:condbeta3q-12} for any $i\in\{1,\ldots,4\}$ there exist $\alpha_i,\beta_i,\alpha_i^{\prime},\beta_i^{\prime}\in\mathbb{F}_{q^4}^*$ such that 
\begin{itemize}
    \item $\beta_i^{q-1}=\beta_0^{q-1}N(\lambda_i)/D(\lambda_i)$;
    \item $(\beta_i^{\prime})^{q-1}=\beta_0^{q-1}N^{\prime}(\lambda_i^{\prime})/D^{\prime}(\lambda_i^{\prime})$;
    \item $\alpha_i=c_i\alpha_0\beta_0 D(\lambda_i)^{q^2+q+1}/N(\lambda_i)^{q+1}$;
    \item $\alpha_i^{\prime}=c_i^{\prime}\alpha_0\beta_0 D^{\prime}(\lambda_i^{\prime})^{q^2+q+1}/N^{\prime}(\lambda_i^{\prime})^{q+1}$,
\end{itemize}
where $c_i,c_i^{\prime}\in\mathbb{F}_q^*$, $N(\lambda_i)=\lambda_i Z^{q^2}Y^q+1$, $N^{\prime}(\lambda_i^{\prime})=\lambda_i^{\prime} (Z^{\prime})^{q^2}(Y^{\prime})^q+1$, $D(\lambda_i)=\lambda_i Z^{q^2}Y+1$ and $D^{\prime}(\lambda_i^{\prime})=\lambda_i^{\prime} (Z^{\prime})^{q^2}Y^{\prime}+1$.

Define the rank-one functions $F_0(x)=\alpha_0 \Tr(\beta_0 x)$, $F(x)=\alpha \Tr(\beta x)$, $F^{\prime}(x)=\alpha^{\prime}\Tr(\beta^{\prime} x)$ and, for $i\in\{1,\ldots,4\}$, $F_i(x):=\frac{1}{c_i}\alpha_i \Tr(\beta_i x)$, $F_i^{\prime}(x):=\frac{1}{c_i^{\prime}}\alpha_i^{\prime} \Tr(\beta_i^{\prime} x)$.
All such functions are elements of the linear $\mathbb{F}_{q^4}$-space $V=\mathcal{G}+\langle F(x),F^{\prime}(x),F_0(x)\rangle_{\mathbb{F}_{q^4}}$ because of Step 1 in Section \ref{Sec:non10}.
We show that, for some suitable choice of the elements $\lambda_i,\lambda_i^{\prime},\alpha_0,\beta_0,\alpha,\beta,\alpha^{\prime},\beta^{\prime}$, the $11$ elements $F(x)$, $F^{\prime}(x)$, $F_0(x)$, $F_1(x),\ldots,F_4(x)$, $F_1^{\prime}(x),\ldots,F_4^{\prime}(x)$ are $\fq$-linearly independent, which implies $\mathrm{trk}(\mathcal{G})=11$.

Let $\mu,\mu^{\prime},\mu_0,\mu_1,\mu_1^{\prime},\ldots,\mu_4,\mu_4^{\prime}\in\fq$ be such that
\[
\mu F(x)+\mu^{\prime}F^{\prime}(x)+\mu_0 F_0(x)+\mu_1 F_1(x)+\cdots \mu_4 F_4(x)+\mu_1^{\prime} F_1^{\prime}(x)+\cdots \mu_4^{\prime} F_4^{\prime}(x)=0,
\]
which can be seen a polynomial identity and hence implies
\begin{equation}\label{eq:sistquasi}
\begin{cases}
\mu\alpha\beta+\mu^{\prime}\alpha^{\prime}\beta^{\prime}+\mu_0\alpha_0\beta_0+\alpha_0\beta_0\sum_{i=1}^{4}\left(\mu_i \frac{D^{q^2+q+1}(\lambda_i)}{N^{q+1}(\lambda_i)}+\mu_i^{\prime} \frac{(D^{\prime})^{q^2+q+1}(\lambda_i)}{(N^{\prime})^{q+1}(\lambda_i)}\right)&=0\\
\mu\alpha\beta^q+\mu^{\prime}\alpha^{\prime}(\beta^{\prime})^q+\mu_0\alpha_0\beta_0^q+\alpha_0\beta_0^q\sum_{i=1}^{4}\left(\mu_i \frac{D^{q^2+q}(\lambda_i)}{N^{q}(\lambda_i)}+\mu_i^{\prime} \frac{(D^{\prime})^{q^2+q}(\lambda_i)}{(N^{\prime})^{q}(\lambda_i)}\right)&=0\\
\mu\alpha\beta^{q^2}+\mu^{\prime}\alpha^{\prime}(\beta^{\prime})^{q^2}+\mu_0\alpha_0\beta_0^{q^2}+\alpha_0\beta_0^{q^2}\sum_{i=1}^{4}\left(\mu_i D^{q^2}(\lambda_i)+\mu_i^{\prime} (D^{\prime})^{q^2}(\lambda_i)\right)&=0\\
\mu\alpha\beta^{q^3}+\mu^{\prime}\alpha^{\prime}(\beta^{\prime})^{q^3}+\mu_0\alpha_0\beta_0^{q^3}+\alpha_0\beta_0^{q^3}\sum_{i=1}^{4}\left(\mu_i N^{q^2}(\lambda_i)+\mu_i^{\prime} (N^{\prime})^{q^2}(\lambda_i)\right)&=0.\\
\end{cases}
\end{equation}
The four equations in \eqref{eq:sistquasi}, together with their images under the $q$-Frobenius map, provide a homogeneous linear system of twelve equations with solutions in $\mathbb{F}_q^{11}$, of which $(\mu,\mu^{\prime},\mu_0,\mu_1,\mu_1^{\prime},\ldots,\mu_4,\mu_4^{\prime})$ is a solution. The matrix $M$ of such a system is
\[
\begin{pmatrix}
YZ & Y^{\prime}Z^{\prime} & 1 & \frac{D^{q^2+q+1}(\lambda_1)}{N^{q+1}(\lambda_1)} & \frac{(D^{\prime})^{q^2+q+1}(\lambda_1^{\prime})}{(N^{\prime})^{q+1}(\lambda_1^{\prime})} & \cdots & \frac{D^{q^2+q+1}(\lambda_4)}{N^{q+1}(\lambda_4)} & \frac{(D^{\prime})^{q^2+q+1}(\lambda_4^{\prime})}{(N^{\prime})^{q+1}(\lambda_4^{\prime})} \\
Y^{q}Z^{q} & (Y^{\prime})^{q}(Z^{\prime})^{q} & 1 & \frac{D^{q^3+q^2+q}(\lambda_1)}{N^{q^2+q}(\lambda_1)} & \frac{(D^{\prime})^{q^3+q^2+q}(\lambda_1^{\prime})}{(N^{\prime})^{q^2+q}(\lambda_1^{\prime})} & \cdots & \frac{D^{q^3+q^2+q}(\lambda_4)}{N^{q^2+q}(\lambda_4)} & \frac{(D^{\prime})^{q^3+q^2+q}(\lambda_4^{\prime})}{(N^{\prime})^{q^2+q}(\lambda_4^{\prime})} \\
Y^{q^2}Z^{q^2} & (Y^{\prime})^{q^2}(Z^{\prime})^{q^2} & 1 & \frac{D^{1+q^3+q^2}(\lambda_1)}{N^{q^3+q^2}(\lambda_1)} & \frac{(D^{\prime})^{1+q^3+q^2}(\lambda_1^{\prime})}{(N^{\prime})^{q^3+q^2}(\lambda_1^{\prime})} & \cdots & \frac{D^{1+q^3+q^2}(\lambda_4)}{N^{q^3+q^2}(\lambda_4)} & \frac{(D^{\prime})^{1+q^3+q^2}(\lambda_4^{\prime})}{(N^{\prime})^{q^3+q^2}(\lambda_4^{\prime})} \\
Y^{q^3}Z^{q^3} & (Y^{\prime})^{q^3}(Z^{\prime})^{q^3} & 1 & \frac{D^{q+1+q^3}(\lambda_1)}{N^{1+q^3}(\lambda_1)} & \frac{(D^{\prime})^{q+1+q^3}(\lambda_1^{\prime})}{(N^{\prime})^{1+q^3}(\lambda_1^{\prime})} & \cdots & \frac{D^{q+1+q^3}(\lambda_4)}{N^{1+q^3}(\lambda_4)} & \frac{(D^{\prime})^{q+1+q^3}(\lambda_4^{\prime})}{(N^{\prime})^{1+q^3}(\lambda_4^{\prime})} \\

Y^{q}Z & (Y^{\prime})^{q}Z^{\prime} & 1 & \frac{D^{q^2+q}(\lambda_1)}{N^{q}(\lambda_1)} & \frac{(D^{\prime})^{q^2+q}(\lambda_1^{\prime})}{(N^{\prime})^{q}(\lambda_1^{\prime})} & \cdots & \frac{D^{q^2+q}(\lambda_4)}{N^{q}(\lambda_4)} & \frac{(D^{\prime})^{q^2+q}(\lambda_4^{\prime})}{(N^{\prime})^{q}(\lambda_4^{\prime})} \\
Y^{q^2}Z^{q} & (Y^{\prime})^{q^2}(Z^{\prime})^{q} & 1 & \frac{D^{q^3+q^2}(\lambda_1)}{N^{q^2}(\lambda_1)} & \frac{(D^{\prime})^{q^3+q^2}(\lambda_1^{\prime})}{(N^{\prime})^{q^2}(\lambda_1^{\prime})} & \cdots & \frac{D^{q^3+q^2}(\lambda_4)}{N^{q^2}(\lambda_4)} & \frac{(D^{\prime})^{q^3+q^2}(\lambda_4^{\prime})}{(N^{\prime})^{q^2}(\lambda_4^{\prime})} \\
Y^{q^3}Z^{q^2} & (Y^{\prime})^{q^3}(Z^{\prime})^{q^2} & 1 & \frac{D^{1+q^3}(\lambda_1)}{N^{q^3}(\lambda_1)} & \frac{(D^{\prime})^{1+q^3}(\lambda_1^{\prime})}{(N^{\prime})^{q^3}(\lambda_1^{\prime})} & \cdots & \frac{D^{1+q^3}(\lambda_4)}{N^{q^3}(\lambda_4)} & \frac{(D^{\prime})^{1+q^3}(\lambda_4^{\prime})}{(N^{\prime})^{q^3}(\lambda_4^{\prime})} \\
YZ^{q^3} & Y^{\prime}(Z^{\prime})^{q^3} & 1 & \frac{D^{q+1}(\lambda_1)}{N(\lambda_1)} & \frac{(D^{\prime})^{q+1}(\lambda_1^{\prime})}{N^{\prime}(\lambda_1^{\prime})} & \cdots & \frac{D^{q+1}(\lambda_4)}{N(\lambda_4)} & \frac{(D^{\prime})^{q+1}(\lambda_4^{\prime})}{N^{\prime}(\lambda_4^{\prime})} \\

Y^{q^2}Z & (Y^{\prime})^{q^2}Z^{\prime} & 1 & D^{q^2}(\lambda_1) & (D^{\prime})^{q^2}(\lambda_1^{\prime}) & \cdots & D^{q^2}(\lambda_4) & (D^{\prime})^{q^2}(\lambda_4^{\prime}) \\
Y^{q^3}Z^{q} & (Y^{\prime})^{q^3}(Z^{\prime})^{q} & 1 & D^{q^3}(\lambda_1) & (D^{\prime})^{q^3}(\lambda_1^{\prime}) & \cdots & D^{q^3}(\lambda_4) & (D^{\prime})^{q^3}(\lambda_4^{\prime}) \\
YZ^{q^2} & Y^{\prime}(Z^{\prime})^{q^2} & 1 &
D(\lambda_1) & D^{\prime}(\lambda_1^{\prime}) & \cdots & D(\lambda_4) & D^{\prime}(\lambda_4^{\prime}) \\
Y^{q}Z^{q^3} & (Y^{\prime})^{q}(Z^{\prime})^{q^3} & 1 & D^{q}(\lambda_1) & (D^{\prime})^{q}(\lambda_1^{\prime}) & \cdots & D^{q}(\lambda_4) & (D^{\prime})^{q}(\lambda_4^{\prime}) \\

Y^{q^3}Z & (Y^{\prime})^{q^3}Z^{\prime} & 1 & N^{q^2}(\lambda_1) & (N^{\prime})^{q^2}(\lambda_1^{\prime}) & \cdots & N^{q^2}(\lambda_4) & (N^{\prime})^{q^2}(\lambda_4^{\prime}) \\
YZ^{q} & Y^{\prime}(Z^{\prime})^{q} & 1 & N^{q^3}(\lambda_1) & (N^{\prime})^{q^3}(\lambda_1^{\prime}) & \cdots & N^{q^3}(\lambda_4) & (N^{\prime})^{q^3}(\lambda_4^{\prime}) \\
Y^{q}Z^{q^2} & (Y^{\prime})^{q}(Z^{\prime})^{q^2} & 1 &
N(\lambda_1) & N^{\prime}(\lambda_1^{\prime}) & \cdots & N(\lambda_4) & N^{\prime}(\lambda_4^{\prime}) \\
Y^{q^2}Z^{q^3} & (Y^{\prime})^{q^2}(Z^{\prime})^{q^3} & 1 & N^{q}(\lambda_1) & (N^{\prime})^{q}(\lambda_1^{\prime}) & \cdots & N^{q}(\lambda_4) & (N^{\prime})^{q}(\lambda_4^{\prime}) \\
\end{pmatrix}.
\]
Since $|\mathbb{F}_q^*|\geq4$, we can choose  $\lambda_1\ne0$, $\lambda_1^2\ne1$ and $\lambda_1^3\ne1$ and then 
\begin{itemize}
    \item $\lambda_2=\lambda_1^2$, $\lambda_3=\lambda_1^3$ and $\lambda_4=\lambda_1^4$;
    \item $\lambda_i=\lambda_{i-4}$ for any $i \in \{5,6,7,8\}$.
\end{itemize}
We also choose $\alpha_0,\beta_0,\alpha,\beta,\alpha^{\prime},\beta^{\prime}$ such that $Z^{q^2+1}=1$ and $Z^\prime=Z^q$.
By direct computation with MAGMA,
\begin{eqnarray*}
\det(M) &=& \lambda_1^{40}(\lambda_1-1)^{12}(\lambda_1+1)^4(\lambda_1^2+\lambda_1+1)^2(Z^2-1)^{6q+6}(Z^q-Z)^4(Z^{q+1}-1)^4 \\
&&  \cdot (Z^{3q+2} - Z^{2q+1} - 2Z^{q+2} + Z^q + Z^3)(Z^{3q+1} - Z^{2q+2} + Z^{q+3} - 2Z^{q+1} + 1)\ \\
&&  \cdot (Z^{3q} + Z^{2q+3} - 2Z^{2q+1} - Z^{q+2} + Z)(Z^{3q+3} - 2Z^{2q+2} + Z^{2q} - Z^{q+1} + Z^2)^2.
\end{eqnarray*}
For some $Z\in\mathbb{F}_{q^4}\setminus\mathbb{F}_{q^2}$ satisfying $Z^{q^2+1}=1$, we have $\det(M)\ne0$; for $q\geq16$ this follows because $q^2+1$ is greater than the sum of the degrees of the polynomials in parentheses, while for $q<16$ this follows by direct checking.
Therefore, for a suitable choice of $\lambda_i,\lambda_i^{\prime},\alpha_0,\beta_0,\alpha,\beta,\alpha^{\prime},\beta^{\prime}$, the matrix $M$ has full rank $11$ and hence $$(\mu,\mu^{\prime},\mu_0,\mu_1,\mu_1^{\prime},\ldots,\mu_4,\mu_4^{\prime})=(0,\ldots,0).$$
Thus, $F(x)$, $F^{\prime}(x)$, $F_0(x)$, $F_1(x)$, $F_1^{\prime}(x)$, $\ldots$, $F_4(x)$, $F_4^{\prime}(x)$ are $\fq$-linearly independent and $\mathcal{G}$ has tensor rank $11$. 
\end{proof}

\section{Tensor rank of $n\times n$ generalized Gabidulin codes for $n\leq4$}\label{Sec:particular}

We compute the tensor rank of some generalized Gabidulin code $\mathcal{C}\subseteq\mathcal{L}_{n,q}$ of dimension $k$ over $\mathbb{F}_{q^n}$ for $n\leq4$.
Notice that, up to equivalence, $\mathcal{C}=\mathcal{G}_{k,1}$.

By Section \ref{Sec:11} and \cite[Table 1]{ByrneCotardo}, the open cases are exactly for $n,k,q$ as in the table below.
Since the lower bound on the tensor rank of $\mathcal{G}_{k,1}$ is $nk+n-k$, we start with the exhaustive search for $t:=n-k$ rank-one functions $\alpha_i \Tr(\beta_i x)\in\mathbb{F}_{q^n}[x]$ such that the rank-one functions in $\mathcal{C}+\langle \alpha_1 \Tr(\beta_1 x),\ldots,\alpha_{t} \Tr(\beta_{t} x) \rangle_{\mathbb{F}_q}$ generate an $\fq$-space $U$ of dimension $nk+t$.
If this succeeds, then we compute explicitly a perfect basis of $U$ (i.e. a basis of pure tensors). Otherwise, we increase $t$ by $1$ and perform the same search again.
In this way we obtain the tensor rank and a perfect basis  $B=\{\eta^i \Tr(\eta^j x) \colon (i,j)\in I\}$ for $\mathcal{G}_{k,1}$, where $\eta$ is a primitive element of $\mathbb{F}_{q^n}$ and $I\subseteq \{0,\ldots,q^n-2\}^2$.
The precise value of the tensor rank is obtained for all but two cases, namely $n=4$, $k=1$ and $q\in\{4,5\}$; in these cases, an upper bound is provided by means of a random search.

\tabcolsep=1 mm
\begin{table}[H]\label{tab}
    \caption{Tensor rank of some generalized Gabidulin codes $\mathcal{G}_{k,1}\subseteq \mathcal{L}_{n,q}$}
    \centering
    \begin{tabular}{|c|c|c|c|c|c|c|}
    \hline
         $n$&$k$&$q$&${\rm TR}(\mathcal{G}_{k,1})$&MTR&${\rm MinPol}(\xi)$&$I$\\
    \hline\hline
         $3$&$2$&$2$&$7$&yes&$x^3+x+1$&(2,0),(3,5),(1,2),(0,4),(6,6),(5,1),(4,3)\\
         \hline
         $4$&$1$&$4$&$\in\{7,8\}$&?&\begin{tabular}{c}$x^8+x^4+$\\$x^3+x^2+1$\end{tabular}& \begin{tabular}{c} (73,168),(22,202),(0,180),(69,249),\\(80,90),(1,96),(33,213),(67,162) \end{tabular} \\
         \hline
         $4$&$1$&$5$&$\in\{7,8\}$&?&\begin{tabular}{c}$x^4+4x^2$\\$+4x+2$\end{tabular}&\begin{tabular}{c} (16,432),(135,81),(21,405),(10,132),\\ (56,593),(24,556),(74,569),(54,268) \end{tabular} \\
         \hline
         $4$&$2$&$2$&$12$&no&$x^4+x+1$&\begin{tabular}{c}
         (14,12),(9,14),(4,0),(7,5),(1,13),(5,10),\\
         (6,3),(2,1),(12,7),(3,2),(0,4),(13,6)
         \end{tabular}\\
         \hline
         $4$&$2$&$3$&$11$&no&$x^4-x^3-1$&\begin{tabular}{c}
         (23,8),(0,13),(28,14),(32,46),(2,1),(19,26),\\
        (1,18),(6,37),(7,12),(36,28),(21,59)\end{tabular}\\
         \hline
         $4$&$2$&$4$&$11$&no&\begin{tabular}{c}$x^8+x^4+$\\$x^3+x^2+1$\end{tabular}&
         \begin{tabular}{c}
         (13,133),(56,175),(20,71),(30,31),(81,124),(0,51),\\
        (3,88),(76,34),(70,215),(29,132),(9,24)\end{tabular}
         \\
         \hline
         $4$&$3$&$2$&$13$&yes&$x^4+x+1$&\begin{tabular}{c}
         (12,13),(8,6),(4,4),(9,4),(1,5),(3,1),(14,9),\\
         (11,0),(10,2),(0,7),(6,10),(7,8),(5,12)
         \end{tabular}\\
         \hline
         $4$&$3$&$3$&$13$&yes&$x^4-x^3-1$&\begin{tabular}{c} (31,3),(29,49),(0,56),(25,61),(7,75),(26,18),(22,30),\\ (20,36),(39,19),(18,2),(13,57),(32,40),(3,47) \end{tabular} \\
         \hline
    \end{tabular}
\end{table}

\begin{remark}
The fourth, fifth, and sixth rows of the table completes the proof of Theorem \ref{Th:main3}.
\end{remark}

\begin{remark}
Notice that, although only one perfect basis is showed in the table, the computations provide a much larger number of perfect bases in each case. Therefore, no generator tensor of such codes is identifiable.
\end{remark}

\section{Appendix}\label{Sec:Appendix}

\begin{theorem}\label{Th:O'Sistemone}
Let $Y,Z\in \mathbb{F}_{q^4}$. Then $(Y,Z)$ is a solution of System \eqref{Eq:Sistemone} if and only if one of the following pairwise mutually exclusive conditions holds:
\begin{enumerate}
    \item[(C1)] $Y\in\mathbb{F}_q$ or $Z\in\fq$;
    \item[(C2)] $Y\notin\mathbb{F}_{q^2}$ and $Z=\rho Y^{q+1}$ for some $\rho\in\fq^*$;
    \item[(C3)] $Z\notin\mathbb{F}_{q^2}$ and $Y=\rho/Z^{q^2+q}$ for some $\rho\in\fq^*$.
\end{enumerate}
\end{theorem}

\begin{proof}
For any $i=0,1,2,3$, write $y_i=Y^{q^i}$ and $z_i=Z^{q^i}$. Then System \eqref{Eq:Sistemone} reads
\begin{equation}\label{Eq:Sistemonecambiovar}
\begin{cases}
f_1(y_0,y_1,y_2,y_3,z_0,z_1,z_2,z_3)=0\\
f_2(y_0,y_1,y_2,y_3,z_0,z_1,z_2,z_3)=0\\
f_3(y_0,y_1,y_2,y_3,z_0,z_1,z_2,z_3)=0,\\
\end{cases}
\end{equation}
where
\begin{eqnarray*}
f_1(y_0,y_1,y_2,y_3,z_0,z_1,z_2,z_3)&=&y_0z_1 - y_0z_2 + y_1z_2 - y_1z_3 - y_2z_0 + y_2z_3 +  y_3z_0 -  y_3z_1,\\
f_2(y_0,y_1,y_2,y_3,z_0,z_1,z_2,z_3)&=&y_0y_1z_1z_2 - y_0y_1z_2z_3 -y_0y_2z_0z_2 + y_0y_2z_1z_3 + y_0y_3 z_0z_1 \\
& &      - y_0y_3 z_1z_2      
 - y_1y_2 z_0z_3 + y_1y_2z_2z_3 + y_1y_3 z_0z_2 -
            y_1y_3 z_1z_3 \\
& &            - y_2y_3 z_0z_1 + y_2y_3 z_0z_3,\\
f_3(y_0,y_1,y_2,y_3,z_0,z_1,z_2,z_3)&=&y_0y_1y_2 z_0z_2z_3 - y_0y_1y_2 z_1z_2z_3 - y_0y_1y_3 z_0z_1z_2 +
            y_0y_1y_3 z_1z_2z_3 \\
& &            +  y_0y_2y_3 z_0z_1z_2 - y_0y_2y_3 z_0z_1z_3 +
            y_1y_2y_3z_0z_1z_3 - y_1y_2y_3 z_0z_2z_3.
\end{eqnarray*}
We denote by ${\rm Res}_x(g_1,g_2)$ the resultant of two (multivariate) polynomials $g_1$ and $g_2$ with respect to the indeterminate $x$.
We have
\begin{eqnarray*}
{\rm Res}_{z_1}\left({\rm Res}_{z_0}(f_3,f_1),{\rm Res}_{z_0}(f_2,f_1)\right)&=&
-(z_2-z_3)^2 y_3 (y_2-y_3)(y_2z_3-y_3z_2)\cdot\\
&&(y_1-y_2)^2(y_0-y_3)(y_0-y_1)^2(y_0z_2-y_3z_3)\cdot\\
&&(y_0z_2-y_2z_3)^2(y_0y_2-y_1y_3)(y_1-y_3).
\end{eqnarray*}
Thus, every solution $(Y,Z)\in\mathbb{F}_{q^4}^2$ of System \eqref{Eq:Sistemone} with $YZ\ne0$ satisfies one of the following conditions.
\begin{enumerate}
    \item $y_3=0$, that is $Y^{q^3}=0$, a contradiction.
    \item $z_2-z_3=0$, that is $Z\in\fq$. Indeed $(Y,Z)$ is a solution of System \eqref{Eq:Sistemone} whenever $Z\in\fq$. In the following cases we can then assume $Z\notin\fq$.
    \item $y_2-y_3=0$, or $y_1-y_2=0$, or $y_0-y_3=0$, or $y_0-y_1=0$. This is equivalent to $Y\in\fq$, and indeed $(Y,Z)$ is a solution of System \eqref{Eq:Sistemone} whenever $Y\in\fq$. In the following cases we can then assume $Y\notin\fq$.
    \item $y_1-y_3=0$, that is $Y\in\mathbb{F}_{q^2}$. Then System \eqref{Eq:Sistemone} reads
    \[
    \begin{cases}
    Z^{q^2+q+1}-Z^{q^3+q+1}+Z^{q^3+q^2+1}-Z^{q^3+q^2+q}=0\\
    Z-Z^q+Z^{q^2}-Z^{q^3}=0\\
    (Y^q+Y)(Z^{q^2+1}-Z^{q^3+q})=0.
    \end{cases}
    \]
    The first equation yields $Z^{q^3+q^2+1}+Z^{q^2+q+1}\in\fq$, that is $Z^{q^2+1}(Z^{q^2}+Z)^{q}\in\fq$, while the second equation yields $Z^{q^2}+Z\in\fq$. Therefore $Z^{q^2+1}\in\fq$, and hence the third equation is also satisfied. Now, the two conditions $Z^{q^2}+Z\in\fq$ and $Z^{q^2+1}\in\fq$ yield $Z^{q^3}=Z^{q^2}-Z^q+Z$ and $Z^{q^3}=Z^{q^2-q+1}$. This implies $Z(Z^{q-1}-1)^{q+1}=0$, whence $Z\in\fq$.
    In the following cases we can then assume $Y\notin\mathbb{F}_{q^2}$.
    \item $y_2z_3-y_3z_2=0$, that is $Y^{q-1}=Z^{q-1}$, and hence $Y=\rho Z$ for some $\rho\in\fq^*$. Then $y_i=\rho z_i$ for any $i=1,\ldots,3$, and System \eqref{Eq:Sistemonecambiovar} reads
    \[
    \begin{cases}
    g_1:=z_0z_1 - 2z_0z_2 + z_0z_3 + z_1z_2 - 2z_1z_3 + z_2z_3=0\\
    g_2:=z_0^2z_1z_3 - z_0^2z_2^2 + z_0z_1^2z_2 - 2z_0z_1z_2z_3 + z_0z_2z_3^2 - z_1^2z_3^2 + z_1z_2^2z_3=0.
    \end{cases}
    \]
    From ${\rm Res}_{z_1}(g_1,g_2)=0$ it follows that $(z_2-z_3)^2(z_3-z_0)^2(z_2-z_0)^2=0$, which is equivalent to $Z\in\mathbb{F}_{q^2}$. Then $Y=\rho Z\in\mathbb{F}_{q^2}$.
    \item $y_0z_2-y_3z_3=0$, that is $Y^{q-1}=1/Z^{(q^2+q+1)(q-1)}$, and hence $Y=\rho /Z^{q^2+q+1}$ for some $\rho\in\mathbb{F}_q^*$. Then $y_i=\rho/(z_iz_{i+1}z_{i+2})$ for any $i=0,\ldots,3$ (where the indices are modulo $4$) and System \eqref{Eq:Sistemonecambiovar} reads
    \[
    \begin{cases}
    h_1:=z_0 z_1 - 2z_0 z_2 + z_0 z_3 + z_1 z_2 - 2 z_1 z_3 + z_2 z_3 = 0 \\
    h_2:=z_0 ^2 z_1 z_3 - z_0 ^2 z_2^2 + z_0 z_1^2 z_2 - 2 z_0 z_1 z_2 z_3 + z_0 z_2 z_3^2 - z_1^2 z_3^2 + z_1 z_2^2 z_3 = 0.\\
    \end{cases}
    \]
    From ${\rm Res}_{z_3}(h_1,h_2)=0$ it follows that $z_2=z_0$.
    Then $Z\in\mathbb{F}_{q^2}$ and hence $Y=\rho/(Z^2 Z^q)$. By System \eqref{Eq:Sistemone}, this implies $Z\in\fq$.
    \item $y_0z_2-y_2z_3=0$, that is $YZ^{q^2}-Y^{q^2}Z^{q^3}=0$. This is equivalent to $Z^{q-1}=Y^{q^2-1}$, and hence to $Z=\rho Y^{q+1}$ for some $\rho\in\mathbb{F}_q^*$. By direct checking, this is indeed a solution of System \eqref{Eq:Sistemone} for any $Y\in\mathbb{F}_{q^4}$. If we require $Y\notin\mathbb{F}_{q^2}$, this also implies $Z\notin\fq$.
    \item $y_0y_2-y_1y_3=0$, that is $y_3=y_0y_2/y_1$, or equivalently $Y^{q^3}=Y^{1+q^2}/Y^q$. Then System \eqref{Eq:Sistemonecambiovar} reads
    \[
    \begin{cases}
    p_1:=y_0 y_1 z_1 - y_0 y_1 z_2 + y_0 y_2 z_0- y_0 y_2 z_1 + y_1^2 z_2 - y_1^2 z_3 - y_1 y_2 z_0+ y_1 y_2 z_3 =0\\
    p_2:=y_0 y_1 z_0 z_1 z_2 - y_0 y_1 z_1 z_2 z_3 - y_0 y_2 z_0 z_1 z_2 + y_0 y_2 z_0 z_1 z_3 - y_1^2 z_0 z_2 z_3 \\\qquad + y_1^2 z_1 z_2 z_3 - y_1 y_2 z_0 z_1 z_3 + y_1 y_2 z_0 z_2 z_3 = 0\\
    p_3:=y_0^2 y_2 z_0 z_1 - y_0^2 y_2 z_1 z_2 + y_0 y_1^2 z_1 z_2 - y_0 y_1^2 z_2 z_3 - y_0 y_2^2 z_0 z_1\\\qquad + y_0 y_2^2 z_0 z_3 -
            y_1^2 y_2 z_0 z_3 + y_1^2 y_2 z_2 z_3 = 0.
    \end{cases}
    \]
    From ${\rm Res}_{y_2}(p_1,p_2)=0$ it follows that $(z_0-z_2)(z_0z_2-z_1z_3)(y_0z_1-y_1z_3)=0$.
    \begin{itemize}
        \item[8.1] Suppose $z_0-z_2=0$, i.e. $Z\in\mathbb{F}_{q^2}$, whence also $z_3=z_1$. Then, by System \eqref{Eq:Sistemonecambiovar}, either $Y\in\fq$ or $Z\in\fq$.
        \item[8.2] Suppose $z_0z_2-z_1z_3=0$, so that $z_3=z_0z_2/z_1$. By System \eqref{Eq:Sistemonecambiovar},
        \[
        \begin{cases}
        \ell_1:= y_0^2  y_2 z_0  z_1 - y_0^2  y_2  z_1  z_2 +  y_0  y_1^2  z_1  z_2 -  y_0  y_1^2  z_2^2 -  y_0  y_2^2 z_0  z_1\\
        \qquad+  y_0  y_2^2 z_0  z_2 -  y_1^2  y_2 z_0  z_2 + y_1^2  y_2  z_2^2 = 0 \\
        \ell_2:=y_0 z_0  z_1 -  y_0  z_1  z_2 -  y_1 z_0  z_2 +  y_1  z_1  z_2 -  y_2 z_0  z_1 +  y_2 z_0  z_2 = 0 \\
        \ell_3:=y_0  y_1  z_1 -  y_0  y_1  z_2 +  y_0  y_2 z_0 -  y_0  y_2  z_1 -  y_1  y_2 z_0 + y_1  y_2 z_2 = 0.
        \end{cases}
        \]
        From ${\rm Res}_{y_2}(\ell_1,\ell_3)=0$ it follows $y_0z_1-y_1z_2=0$, so that $Y^{q-1}=(1/Z^q)^{q-1}$. This implies $Y=\rho/Z^q$ for some $\rho\in\fq^*$, whence $y_i=\rho/z_{i+1}$ for any $i=0,\ldots,3$ (indices modulo $4$). Then, by  System \eqref{Eq:Sistemonecambiovar},
        \[
        \begin{cases}
        m_1:=z_0^2 z_1 - z_0 z_1 z_2 - 2 z_0 z_1 z_3 + z_0  z_2 z_3 + z_1^2 z_3 = 0 \\
        
        m_2:=z_0^2 z_2 + z_0 z_1^2 - 2 z_0 z_1 z_2 - z_0 z_1 z_3 + z_1 z_2 z_3 = 0 \\

        m_3:=z_0^2 z_1 + z_0^2 z_2 + z_0 z_1^2 - 3 z_0 z_1 z_2 - 3 z_0 z_1 z_3 + z_0 z_2 z_3 + z_1^2 z_3 + z_1 z_2 z_3 = 0.
        \end{cases}
        \]
        From ${\rm Res}_{z_3}(m_1,m_2)=0$ it follows that $Z\in\fq$.
        \item[8.3] Suppose $y_0z_1-y_1z_3=0$. This implies $Y^{q-1}=(1/Z^{q^2+q})^{q-1}$, whence $Y=\rho/Z^{q^2+q}$ for some $\rho\in\fq^*$. Then $y_i=\rho/(z_{i+2}z_{i+1})$ for any $i=0,\ldots,3$ (indices modulo $4$), which is indeed a solution of System \eqref{Eq:Sistemonecambiovar} and provide a solution $(Y,Z)$ of System \eqref{Eq:Sistemone}. Notice that the condition $Y=\rho/Z^{q^2+q}$ with $Z\in\mathbb{F}_{q^4}$ implies $Y^{q^3}=Y^{1+q^2}/Y^q$. 
        For such a solution, the require $Z\notin\mathbb{F}_{q^2}$ is equivalent to $Y\notin\fq$. 
        Also, if $Y,Z\in\mathbb{F}_{q^4}^*$ are such that $Y=\rho/Z^{q^2+q}$ and $Z=\rho^\prime Y^{q+1}$ with $\rho,\rho^{\prime}\in\fq$, then $Y^{q^3+q^2+q+1+q^2}=\rho/\rho^\prime\in\fq^*$, whence $Y\in\fq$.
    \end{itemize}
\end{enumerate}
\end{proof}

\begin{theorem}\label{Th:Ranghi}
Let $(Y,Z)\in\mathbb{F}_{q^4}^2$ be a solution of System \eqref{Eq:Sistemone}, and $M$ be the matrix in \eqref{eq:matrice}.
\begin{itemize}
    \item[(R1)] If $Y\in\fq$ or $Z\in\fq$, then ${\rm rank}(M)=2$.
    \item[(R2)] If $Y\notin\mathbb{F}_{q^2}$ and $Z=\rho Y^{q+1}$ for some $\rho\in\fq^*$, then ${\rm rank}(M)=6$.
    \item[(R3)] If $Z\notin\mathbb{F}_{q^2}$ and $Y=\rho/ Z^{q^2+q}$ for some $\rho\in\fq^*$, then ${\rm rank}(M)=6$. 
\end{itemize}
\end{theorem}

\begin{proof}
For any $i=1,\ldots,10$, denote respectively by $M^{(i)}$ and $M_{(i)}$ the $i$-th row and the $i$-th column of $M$. Note that $M_{(9)}$ and $M_{(10)}$ are linearly independent.
Note also that, by construction of $M$, any possible $\fq$-linear combination of the columns of $M$ needs to be checked only on the first, second, third, and seventh rows of $M$.
\begin{itemize}
    \item[(R1)] Suppose $Y\in\fq$. Then $N(\lambda)=D(\lambda)$, whence $\frac{D^{q^2+q+1}(\lambda)}{N^{q+1}(\lambda)}=\frac{D^{q^2+q}(\lambda)}{N^{q}(\lambda)}=\lambda ZY+1$. Therefore $M_{(j)}=\lambda_j M_{(9)}+M_{(10)}$ for any $j=1,\ldots,8$, and ${\rm rank}(M)=2$.
    
    Suppose $Z\in\fq$. Similarly, one has $N(\lambda)=D^q(\lambda)$ and $M_{(j)}=\lambda_j M_{(9)}+M_{(10)}$ for any $j=1,\ldots,8$, so that ${\rm rank}(M)=2$.
\end{itemize}
For $i=6,7$, denote by $S_i$ the $i\times i$ submatrix of $M$ given by the first $i$ rows and the last $i$ columns of $M$.
If for any distinct $\lambda_4,\ldots,\lambda_8\in\fq^*$
one has $\det(S_6)\ne0$ and $\det(S_7)=0$, then this is enough to conclude that ${\rm rank}(M)=6$ for any distinct $\lambda_1,\ldots,\lambda_8\in\fq^*$. In fact, the use of the column $M_{(j)}$, $j\in\{1,2,3\}$, instead of $M_{(4)}$, implies the replacement of $\lambda_4$ with $\lambda_j$ in $\det(S_7)$, and in this way the fourth column becomes any of the remaining columns.  
If $S_7^{\prime}$ is obtained from $S_7$ by replacing $M^{(7)}$ with $M^{(i)}$, $i\in\{8,9,10\}$, then the elementwise $q^{i-7}$-power $\Phi$ maps $M^{(7)}$ to $M^{(i)}$, while $M^{(1)}$ and $M^{(2)}$ are fixed by $\Phi$, and $M^{(3)},M^{(4)},M^{(5)},M^{(6)}$ are cyclically permuted by $\Phi$.
Therefore $\Phi$ maps the rows of $S_7$ to the rows of $S_7^{\prime}$, so that $\det(S_7)=0$ if and only if $\det(S_7^{\prime})=0$.
\begin{itemize}
\item[(R2)] Suppose $Y\notin\mathbb{F}_{q^2}$ and $Z=\rho Y^{q+1}$ for some $\rho\in\fq^*$. Then $\det(S_6)$ equals
\begin{eqnarray*}
\rho^{10}Y^{4(q^3+q^2+q+1)}\left(\prod_{i=5}^8\lambda_i\right)\left(\prod_{5\leq i<j\leq8}(\lambda_i-\lambda_j)\right)\cdot\\
\left(\prod_{i=0}^3(Y^{q^i}-Y^{q^{i+1}})^2\right)\left(\prod_{i=0}^1(Y^{q^i}-Y^{q^{i+2}})^3\right),
\end{eqnarray*}
and hence $\det(S_6)\ne0$ because $Y\notin\mathbb{F}_{q^2}$ and the $\lambda_i$'s are nonzero and distinct. Also, $\det(S_7)=0$. Therefore, ${\rm rank}(M)=6$.
\item[(R3)] Suppose $Z\notin\mathbb{F}_{q^2}$ and $Y=\rho/ Z^{q^2+q}$ for some $\rho\in\fq^*$. Then
\begin{eqnarray*}
\det(S_6)=\rho^{10}\left(\prod_{i=5}^8\lambda_i\right)\left(\prod_{5\leq i<j\leq8}(\lambda_i-\lambda_j)\right)\cdot\\
\left(\prod_{i=0}^3(Z^{q^i}-Z^{q^{i+1}})^2\right)\left(\prod_{i=0}^1(Z^{q^i}-Z^{q^{i+2}})^3\right)
\end{eqnarray*}
is nonzero. Also, $\det(S_7)=0$. Therefore, ${\rm rank}(M)=6$.
\end{itemize}
\end{proof}

\section*{Acknowledgments}

The authors of this paper would like to thank Alessandro Neri for fruitful discussions.
This research was supported by the Italian National Group for Algebraic and Geometric Structures and their Applications (GNSAGA - INdAM). The last two authors were supported by the project ``VALERE: VAnviteLli pEr la RicErca" of the University of Campania ``Luigi Vanvitelli''.

\end{document}